\documentclass[12pt,reqno]{article}

\usepackage[usenames]{color}
\usepackage{amssymb}
\usepackage{amsmath}
\usepackage{amsthm}
\usepackage{amsfonts}
\usepackage{amscd}
\usepackage{graphicx}

\usepackage[colorlinks=true,
linkcolor=webgreen,
filecolor=webbrown,
citecolor=webgreen]{hyperref}

\definecolor{webgreen}{rgb}{0,.5,0}
\definecolor{webbrown}{rgb}{.6,0,0}

\usepackage{color}
\usepackage{fullpage}
\usepackage{float}
\usepackage[shortlabels]{enumitem}

\usepackage{graphics}
\usepackage{latexsym}
\usepackage{epsf}
\usepackage{breakurl}

\newcommand{\seqnum}[1]{\href{https://oeis.org/#1}{\rm \underline{#1}}}

\DeclareMathOperator{\per}{per}
\DeclareMathOperator{\ice}{ice}

\DeclareMathOperator{\nnp}{nnp}

\DeclareMathOperator{\PER}{PER}
\DeclareMathOperator{\pdp}{pdp}
\DeclareMathOperator{\pdlp}{pdlp}
\def\suchthat{ \, : \, }

\begin{document}

\theoremstyle{plain}
\newtheorem{theorem}{Theorem}
\newtheorem{corollary}[theorem]{Corollary}
\newtheorem{lemma}[theorem]{Lemma}
\newtheorem{proposition}[theorem]{Proposition}

\theoremstyle{definition}
\newtheorem{definition}[theorem]{Definition}
\newtheorem{example}[theorem]{Example}
\newtheorem{conjecture}[theorem]{Conjecture}

\theoremstyle{remark}
\newtheorem{remark}[theorem]{Remark}

\title{An inequality for the number of periods in a word}

\author{Daniel Gabric\footnote{School of Computer Science, University of Waterloo, Waterloo, ON  N2L 3G1, Canada; {\tt dgabric@uwaterloo.ca}; {\tt shallit@uwaterloo.ca}.},\quad Narad Rampersad\footnote{
Department of Math/Stats,
University of Winnipeg,
515 Portage Ave.,
Winnipeg, MB, R3B 2E9
Canada; {\tt narad.rampersad@gmail.com}.}
,\quad and Jeffrey Shallit$^*$}

\maketitle

\begin{abstract}
We prove an inequality for the number of periods in a word $x$
in terms of the length of $x$ and its initial critical exponent.  Next, we characterize all periods of the length-$n$ prefix of a characteristic Sturmian word in terms of the lazy Ostrowski representation of $n$, and use this result to show that our
inequality is tight for infinitely many words
$x$.  We propose two related measures of periodicity for infinite words.   Finally,  we also consider special cases where $x$ is overlap-free or
squarefree.    
\end{abstract}

\section{Introduction}

Let $x$ be a finite nonempty word of length $n$.  We say that
an integer $p$, $1 \leq p \leq n$, is a {\it period\/} of
$x$ if $x[i] = x[i+p]$ for $1 \leq i \leq n-p$.   For
example, the English word {\tt alfalfa} has periods
$3,6,$ and $7$.  A
period $p$ is {\it nontrivial\/} if $p < n$; the period $n$ is
{\it trivial\/} and is often ignored.
The least period of a word is sometimes
called {\it the\/} period and is written
$\per(x)$.  The number of nontrivial periods of a word
$x$ is written $\nnp(x)$.   Sometimes the prefix
$x[1..p]$ is also called a period; in general, this should cause no confusion.

The {\it exponent\/} of a length-$n$ word $x$ is defined
to be $\exp(x) = n/\per(x)$.   For example, the
French word {\tt entente} has exponent $7/3$.   The
{\it initial critical exponent\/} $\ice(x)$ of a finite or infinite word $x$ is defined to be 
$$ \ice(x) := \sup_{{p \text{ a nonempty}} \atop {\text{ prefix of $x$}}} \exp(p).$$
For example, $\ice({\tt phosphorus}) = 7/4$.  
This concept was (essentially) introduced by
Berth\'e, Holton, and Zamboni \cite{Berthe&Holton&Zamboni:2006}.

A word $w$ is a {\it border\/} of $x$ if $w$ is both
a prefix and a suffix of $x$.   Although overlapping borders
are allowed, by convention we generally rule out
borders $w$ where $|w| \in \{ 0, |x| \}$.

There is an obvious relationship between borders and periods:
a length-$n$ word $x$ has a nontrivial period $t$ iff it has a border
of length $n-t$.
For example, the English word {\tt abracadabra}
has periods $7, 10$, and $11$, and borders of length $1$ and
$4$.

A word is {\it unbordered} if it has no borders and {\it bordered\/}
otherwise.   An unbordered 
word $x$ has only the trivial period $|x|$.   On the other
hand, a word of the form $a^n$, for $a$ a single letter,
evidently has the largest possible number of periods; namely,
$n$.

In this note we prove an inequality that gives an upper bound
for $\nnp(x)$,
the number of nontrivial periods of (and hence, the number of
borders in) a word $x$.  Roughly speaking, this inequality says
that, in order for a word to have many periods, it must either
be very long, or have a large initial critical exponent.   
We also prove that our
inequality is tight, up to an additive constant.
To do so, in Section~\ref{three} we characterize all periods of the length-$n$ prefix of a characteristic Sturmian word in terms of the lazy Ostrowski representation of $n$.  In Section~\ref{twomeas}, we propose two related measures of periodicity for infinite words, and we compute these measure for some famous words.   Finally, in the last two sections, we consider the shortest binary overlap-free (resp., ternary squarefree) words having $n$ periods.

\section{The period inequality}

\begin{theorem}
Let $x$ be a bordered word of length $n \geq 1$.
Let $e = \ice(x)$.    Then
\begin{equation}
\nnp(x) \leq {e \over 2 } + 1 + {{\ln (n/2)} \over {\ln (e/(e-1))}} .
\label{bound1}
\end{equation}
\end{theorem}

\begin{proof}
We break the bound up into two pieces, by considering
the periods of size $\leq n/2$ and $> n/2$.  We call
these the {\it short\/} and {\it long\/} periods.

Let $p = \per(x)$, the shortest period of $x$.
If $p$ is short, then
$x$ has short periods $p, 2p, 3p, \ldots, \lfloor n/(2p) \rfloor p$.
Clearly $\ice(x) \geq n/p$, so we get at most
$e/2$ periods from this list.   To see that there are 
no other short periods, let $q$ be some short period
not on this list.  Then $p < q \leq n/2$ by assumption.
By the Fine-Wilf theorem \cite{Fine&Wilf:1965},
if a word of length $n$
has two periods $p, q$ with $n \geq p + q - \gcd(p,q)$, then
it also has period $\gcd(p,q)$.   Since $\gcd(p,q) \leq p$,
either $\gcd(p,q) < p$, which is a contradiction, or
$\gcd(p,q) = p$, which means $q$ is a multiple of $p$,
another contradiction.

Next, let's consider the long periods or, alternatively,
the short borders (those of length $< n/2$).
Suppose $x$ has borders $y, z$
of length $q$ and $r$ respectively, with $q < r < n/2$.
Then $x = y y' y = z z' z$ for words $y'$ and $z'$.
Hence $z = yt = t'y$ for some nonempty words $t$
and $t'$.   Then by the Lyndon-Sch\"utzenberger
theorem (see, e.g., \cite{Lyndon&Schutzenberger:1962})
we know there exist words $u, v$ with $u$ nonempty,
and an integer $d \geq 0$,
such that $t' = uv$, $t = vu$, and $y = (uv)^d u$.
Hence $x$ has the prefix $z = yt = (uv)^{d+1} u$,
which means $e = \ice(x) \geq |z|/|uv| = r/(r-q)$.

Now the inequality $r/(r-q) \leq e$ is equivalent to
$r/q \geq e/(e-1)$.  Thus if $b_1 <  b_2 < \cdots < b_t$
are the lengths of all
the short borders of $x$, by the previous
paragraph we have
$$b_1 \geq 1,\ b_2 \geq (e/(e-1))b_1 \geq e/(e-1),$$
and so forth, and hence $b_t \geq (e/(e-1))^{t-1}$.  
All these borders are of length at most $n/2$,
so $n/2 > b_t \geq (e/(e-1))^{t-1}$.  
Hence
$$t \leq 1 + {{\ln(n/2)} \over {\ln(e/(e-1))}} ,$$
and the result follows.
\end{proof}

It is also possible to simplify the statement of
the bound \eqref{bound1}, at the cost of being less precise.

\begin{corollary}
Let $x$ be a word of length $n \geq 1$, and let $e = \ice(x)$.
Then
\begin{enumerate}[(a)]
\item $\nnp(x) \leq {e \over 2 } + 1 + (e - {1 \over 2}) \ln (n/2) $;
\item $\nnp(x) \leq C e \ln n$, 
where $C = 3/(2 \ln 2) \doteq 2.164$.
\end{enumerate}
\label{bound2}
\end{corollary}

\begin{proof}
\begin{enumerate}[(a)]

\item
Start with \eqref{bound1}.  If $e > 1$,
then by computing the Taylor series for
${1 \over {\ln(e/(e-1))}}$,
we see that
$${1 \over {\ln(e/(e-1))}} \leq e - {1 \over 2} .$$

If $e = 1$, then $x$ is unbordered.
The left-hand side of (a) is then $0$, while the right-hand side
is at least $3/2 + (1/2) \ln n/2 \geq 1$.

\item
If $n = 1$ then the desired inequality follows trivially.

Otherwise assume $n \geq 2$.   It is easy to check that
$$ 1 + {1\over2} \ln 2 = (\ln 2 - {1 \over 2}) + {1\over 2}\ln 2 + (C-1)\ln 2$$
where $C = 3/(2 \ln 2)$.  Thus
$$ 1 + {1 \over 2} \ln 2 \leq (\ln 2 - {1 \over 2})e + {1 \over 2}\ln n
	+ (C-1) e \ln n,$$
since $n \geq 2$ and $e \geq 1$.
Now add $e \ln n$ to both sides and rearrange to get
$$ {e \over 2} + 1 + (e - {1 \over 2}) \ln (n/2) \leq C e \ln n,$$
which by (a) gives the desired result.
\end{enumerate}
\end{proof}

It is natural to wonder how tight the bound \eqref{bound1} is for a ``typical'' word of length $n$.  The following two results imply that the expected value of the left-hand side of \eqref{bound1} is $O(1)$, while the expected value of the right-hand side is $\Theta(\ln n)$.  Our inequality, therefore, implies
nothing useful about the ``typical'' word.

\begin{theorem}
Let $k \geq 2$.
Over a $k$-letter alphabet, the
expected number of borders (or the number of nontrival periods)
of a length-$n$ word is $k^{-1} + k^{-2} + \cdots + k^{1-n} 
\leq {1 \over {k-1}}$.
\end{theorem}

\begin{proof}
By the linearity of expectation, the expected number of borders is the sum, from $i = 1$ to $n-1$, of the expected value of the indicator random variable $B_i$ taking the value 1 if there is a border of length $i$,
and $0$ otherwise.
Once the left border of length $i$ is chosen arbitrarily,
the $i$ bits of the
right border are fixed,
and so there are $n-i$ free choices of symbols.  This means that
$E[B_i] = k^{n-i}/k^n = k^{-i}$.
\end{proof}

\begin{theorem}
The expected value of
$\ice(x)$, for finite or infinite words $x$, is $\Theta(1)$.
\label{expbord}
\end{theorem}

\begin{proof}
Let's count the  fraction $H_j$ of words having at least a $j$'th
power prefix.
Count the number of words having a $j$'th power prefix with period 1, 2, 3, etc.
This double counts, but shows that
$H_j \leq k^{1-j} + k^{2(1-j)} + \cdots = 1/(k^{j-1} - 1)$
for $j \geq 2$.
Clearly $H_1 = 1$.

Then $H_{j-1} - H_j$ is the fraction of words having a $(j-1)$th power prefix
but no $j$th power prefix.  These words will have an ice at most $j$.
So the expected value of ice is bounded above by
\begin{align*}
2(H_1 - H_2) + 3(H_2 - H_3) + 4(H_3 - H_4) + \cdots 
&= 2 H_1 + H_2 + H_3 + H_4 + \cdots \\
&= 2 + H_2 + H_3 + H_4 + \cdots \\
&= 2 + \sum_{j \geq 2} 1/(k^{j-1} - 1) \\
&= 2 + \sum_{j \geq 1} 1/(k^j - 1).
\end{align*}
\end{proof}

\section{Periods of prefixes of characteristic Sturmian words}
\label{three}

In this section we take a brief digression to completely characterize the periods of the length-$n$ prefix
of the characteristic Sturmian word with slope $\alpha$.
This characterization is based on a remarkable
connection between these periods and the
so-called
``lazy Ostrowski'' representation of $n$.   
Theorem~\ref{ostt} below implies that all the periods of a length-$n$ prefix of a Sturmian
characteristic word can be read off directly from the lazy Ostrowski representation
of $n$.  

We start by recalling the Ostrowski numeration system.
Let $0 < \alpha < 1$ be an irrational real number with
continued fraction expansion $[0, a_1, a_2, \ldots ]$.
Define $p_i/q_i$ to be the $i$'th convergent to this
continued fraction, so that
$[0, a_1, a_2, \ldots, a_i] = p_i/q_i$.  In the
(ordinary) Ostrowski numeration system, we write every positive
integer in the form
\begin{equation}
n = \sum_{0 \leq i \leq t} d_i q_i,
\label{ost}
\end{equation}
where $d_t > 0$ and the $d_i$ have to obey three conditions:
\begin{enumerate}[(a)]
\item $0 \leq d_0 < a_1$;
\item $0 \leq d_i \leq a_{i+1}$ for $i \geq 1$;
\item For $i \geq 1$, if $d_i = a_{i+1}$ then $d_{i-1} = 0$.
\end{enumerate}
See, for example, \cite[\S 3.9]{Allouche&Shallit:2003}.

The {\it lazy Ostrowski representation} is again
defined through the sum \eqref{ost}, but with 
slightly different conditions:
\begin{enumerate}[(a)]
\setcounter{enumi}{3}
\item $0 \leq d_0 < a_1$;
\item $0 \leq d_i \leq a_{i+1}$ for $i \geq 1$;
\item For $i \geq 2$, if $d_i = 0$, then $d_{i-1} = a_i$;
\item If $d_1 = 0$, then $d_0 = a_i - 1$.
\end{enumerate}
See, for example,
\cite[\S 5]{Epifanio&Frougny&Gabriele&Mignosi&Shallit:2012}.
By convention, the Ostrowski representation is written as a finite word $d_t d_{t-1} \cdots d_1 d_0$,
starting with the most significant digit.

Next, we recall the definition of the characteristic Sturmian
infinite word ${\bf x}_\alpha = x_1 x_2 x_3 \cdots$.  It is defined by
$$ x_i = \lfloor (i+1) \alpha \rfloor - \lfloor i \alpha \rfloor$$
for $i \geq 1$.    For more about Sturmian words, see
\cite{Berstel&Seebold:2002,Reutenauer:2019,Berstel&Lauve&Reutenauer&Saliola:2009}.

\begin{example}
Take $\alpha = \sqrt{2} - 1 = [0,2,2,2,\ldots]$.   
Then $q_0 = 1$, $q_1 = 2$, $q_2 = 5$, $q_3 = 12$.   The first few
ordinary and lazy Ostrowski representations are given in the table below.
\begin{center}
\begin{tabular}{c|c|c||c|c|c}
$n$ & ordinary & lazy & $n$ & ordinary & lazy  \\
 & Ostrowski & Ostrowski  & & Ostrowski & Ostrowski \\
 \hline
 1 & 1 & 1  & 15 & 1011 & 221 \\
 2 & 10 & 10 & 16 & 1020 & 1020 \\
 3 & 11 & 11 & 17 & 1100 & 1021 \\
 4 & 20 & 20 & 18 & 1101 & 1101 \\
 5 & 100 & 21 & 19 & 1110 & 1110 \\
 6 & 101 & 101 & 20 & 1111 & 1111 \\
 7 & 110 & 110 & 21 & 1120 & 1120 \\
 8 & 111 & 111 & 22 & 1200 & 1121 \\
 9 & 120 & 120 & 23 & 1201 & 1201 \\
 10 & 200 & 121 & 24 & 2000 & 1210 \\
 11 & 201 & 201 & 25 & 2001 & 1211 \\
 12 & 1000 & 210 & 26 & 2010 & 1220 \\
 13 & 1001 & 211 & 27 & 2011 & 1221 \\
 14 & 1010 & 220 & 28 & 2020 & 2020 
\end{tabular}
\end{center}
\end{example}

In what follows, fix a suitable $\alpha$.
Let $Y_n$ for $n \geq 1$ be the prefix of
${\bf x}_\alpha$ of length $n$, and define
$X_n := Y_{q_n}$.  Let
$\PER(n)$ denote the set of all
periods of $Y_n$ (including the trivial period $n$).  Then we have the following result, which gives a complete characterization of the periods of $Y_n$.  It can be viewed as a generalization of a 2009 theorem of Currie and Saari \cite[Corollary 8]{Currie&Saari:2009}, which obtained the least period of $X_n$.
\begin{theorem}\label{ostt}
\leavevmode
\begin{enumerate}[(a)]
\item The number of periods of $Y_n$ (including the trivial period $n$)
is equal to the sum of the digits in the lazy Ostrowski 
representation of $n$.

\item Suppose the lazy Ostrowski representation of $n$
is $\sum_{0 \leq i \leq t} d_i q_i$.   Define
$$A(n) = \left\lbrace e q_j + \sum_{j < i \leq t} d_i q_i : 1 \leq e \leq d_j \text{ and }
0 \leq j \leq t \right\rbrace.$$  Then $\PER(n) = A(n)$.
\end{enumerate}
\end{theorem}

Part (a) follows immediately from part (b), so it suffices to prove (b) alone.
We need some preliminary lemmas.

\begin{lemma}\label{lazy_len}
The lazy Ostrowski representation of $n$ has length $t+1$
if and only if $$q_t+q_{t-1}-1 \leq n \leq q_{t+1}+q_{t}-2.$$
\end{lemma}
\begin{proof}
The largest integer $N$ represented by a lazy Ostrowski
representation of length $t+1$ is the one
where the coefficient of each $q_i$ takes the
maximum possible values allowed by
conditions (d) and (e) above, but ignoring
condition (f); namely
$N = a_1 - 1 + \sum_{1 \leq i \leq t} a_{i+1} q_i$.
Suppose $t$ is even; an analogous proof works for
the case of $t$ odd.   Then
\begin{align*}
q_{t+1} &= a_{t+1} q_t + q_{t-1} \\
q_{t-1} &= a_{t-1} q_{t-2} + q_{t-3} \\
& \quad \vdots \\
q_1 &= a_1 q_0 + 0, 
\end{align*}
which, by telescoping cancellation, gives
\begin{equation}
q_{t+1} = a_{t+1} q_t + a_{t-1} q_{t-2} + \cdots+
a_1 q_0 .
\label{tel1}
\end{equation}
Similarly
\begin{align*}
q_t &= a_t q_{t-1} + q_{t-2} \\
q_{t-2} &= a_{t-2} q_{t-3} + q_{t-4} \\
&\quad \vdots\\
q_2 &= a_2 q_1 + q_0,
\end{align*}
which, by telescoping cancellation, gives
\begin{equation}
q_t = a_t q_{t-1} + a_{t-2} q_{t-3} + \cdots + 
a_2 q_1 + q_0.
\label{tel2}
\end{equation}
Adding Eqs.~\eqref{tel1} and \eqref{tel2} gives
$q_t + q_{t+1} =
1 + a_1 q_0 + \sum_{1 \leq i \leq t} a_{i+1} q_i $,
and hence $N = q_t + q_{t+1} - 2$, as desired.
\end{proof}

\begin{lemma}\label{AsubPER}
We have $A(n) \subseteq \PER(n)$.
\end{lemma}
\begin{proof}
Frid \cite{Frid:2018} defined two kinds of representations in the Ostrowski system.
A representation $n = \sum_{0 \leq i \leq t} d_i q_i$ is {\it legal\/} if $0 \leq d_i \leq a_{i+1}$.
A representation $n = \sum_{0 \leq i \leq t} d_i q_i$
is {\it valid\/} if 
$Y_n = X_t^{d_t} \cdots X_0^{d_0}$.    She
proved \cite[Corollary 1, p.~205]{Frid:2018} that every legal representation is valid.  Since the lazy Ostrowski representation is legal 
\cite[Thm.~47]{Epifanio&Frougny&Gabriele&Mignosi&Shallit:2012}, it follows that if
$n = \sum_{0 \leq i \leq t} d_i q_i$ is the lazy
Ostrowski representation of $n$, then
$Y_n = X_t^{d_t} \cdots X_0^{d_0}$.

We now argue that (thinking of each $X_i$ as a single
symbol) that every nonempty prefix of $X_t^{d_t} \cdots X_0^{d_0}$ is a period of $Y_n$.   In other words, 
\begin{align}
& X_t,\ X_t^2,\ \ldots,\ X_t^{d_t},  \nonumber \\
& X_t^{d_t} X_{t-1},\ X_t^{d_t} X_{t-1}^2,\ \ldots,\  X_t^{d_t} X_{t-1}^{d_{t-1}},  \nonumber \\
& \ldots,  \label{periodlist} \\
& X_t^{d_t} X_{t-1}^{d_{t-1}} \cdots X_1^{d_1} X_0,\ 
X_t^{d_t} X_{t-1}^{d_{t-1}} \cdots X_1^{d_1} X_0^2 ,\ 
\ldots, \ 
 X_t^{d_t} X_{t-1}^{d_{t-1}} \cdots X_1^{d_1} X_0^{d_0}.
\nonumber
\end{align}
are all periods of $Y_n$.

We first handle the periods in the first line of
\eqref{periodlist}, which are all powers of $X_t$.  Note that
every nonempty suffix of a lazy representation is also
lazy, and hence from
Lemma~\ref{lazy_len} we know that
$|X_{t-1}^{d_{t-1}} \cdots X_0^{d_0}| \leq q_t +  q_{t-1} -2
= |X_t X_{t-1}| - 2$.   Furthermore every lazy representation
is valid, so
$Y_n = X_t^{e_t} Z$, where
$Z = Y_{n-e_tq_t}$ is a (possibly empty)
prefix of $X_t X_{t-1}$.
Then $Y_n = X_t^{e_t} Z$ is a prefix
of $X_t^{e_t} X_t X_{t-1}$, which is a prefix
of $X_t^{e_t+2}$, which has period $X_t^j$ for
$0 \leq j\leq e_t$.

Next, we handle the remaining periods, if
there are any.   The next
one in the list \eqref{periodlist} to consider is $X_t^{d_t} X_r$, where $r$ is the largest index $< t$ satisfying
$d_r > 0$.  Thus $Y_n = X_t^{d_t} X_r Z'$, where
$Z' = Y_{n- d_t q_t - q_r}$.  There are two cases to
consider:
\begin{itemize}
\item If $r = t-1$, then 
$ X_r Z' = X_{t-1}^{d_{t-1}} \cdots X_0^{d_0} $
and hence, as above $|X_r Z'| \leq q_t + q_{t-1} - 2$.
It follows that 
$|X_t^{d_t} X_r| = d_t q_t + q_{t-1} \geq q_t + q_{t-1} > q_t + q_{t-1} - 2 \geq |Z'|$.  

\item If $r \leq t-2$, then 
$$ |X_t^{d_t} X_r | = d_t q_t + q_r \geq q_t = a_t q_{t-1} + q_{t-2} \geq q_{t-1} + q_{t-2} > q_{t-1} + q_{t-2} - 2 \geq |X_{r-1}^{d_{r-1}} \cdots X_0^{d_0}|,$$
where in the last step we have used Lemma~\ref{lazy_len} again.
\end{itemize}

Hence in both cases the next period in the list is
of size greater than $n/2$, and hence so is every
period following it in the list.   Thus for every period $P$ after
the first line we have $Y_n = P Z'$ where
$|P|>|Z'|$.   Since $Z'$ is also a valid Ostrowski representation of $n - |P|$, it follows that 
$Z' = Y_{n- |P|}$ is a prefix of $P$.  Thus
$Y_n$ has period $P$, as desired.  
\end{proof}

\begin{lemma}\label{min_period}
If $q_t+q_{t-1}-1 \leq n \leq q_{t+1}+q_{t}-2$ then
the smallest period of $Y_n$ is at least $q_t$.
\end{lemma}

\begin{proof}
It suffices to prove the result for $n=q_t+q_{t-1}-1$, since any period of $Y_{n'}$, $n'>n$, is at least as large as
the smallest period of $Y_n$.  Write $Y_{n+1} = X_tX_{t-1}$, where $|X_t|=q_t$ and $|X_{t-1}|=q_{t-1}$.
Let $ab$ be the last two symbols of $X_{t-1}$.  Then $a \neq b$ and we have the well-known ``almost commutative'' property:
$Y_{t-1} = X_tX_{t-1}(ab)^{-1} = X_{t-1}X_t(ba)^{-1}$.  Consequently, the word $Y_{n-1}$ is a \emph{central word}
and has periods $q_t$ and $q_{t-1}$, with $q_{t-1}$ being its smallest period \cite[Proposition~1]{Carpi&deLuca:2005}.
Since $X_{t-1}$ is a prefix of $X_t$, it is clear that $Y_n$ has period $q_t$.  The word $Y_n$ does not have period
$q_{t-1}$, since it would then be a word of length $q_t+q_{t-1}-1$ with co-prime periods $q_t$ and $q_{t-1}$,
contrary to the Fine-Wilf theorem.  The word $Y_n$ therefore does not have any period that is a multiple of $q_{n-1}$.  Furthermore, if $Y_n$ had a period $q$ with
$q_{t-1} < q < q_t$ and $q$ not a multiple of $q_{n-1}$, then  the central word $Y_{n-1}$ would have period $q$ as
well.  The word $Y_{n-1}$ would then have periods $q$ and $q_{t-1}$, again violating the
Fine-Wilf theorem.  It follows that $Y_n$ has smallest
period $q_t$.
\end{proof}

\begin{lemma}\label{PERsubA}
We have $\PER(n) \subseteq A(n)$.
\end{lemma}

\begin{proof}
The proof is by induction on $n$.  Certainly the result holds for $n=1$.
Suppose the lazy Ostrowski representation of $n$
is $\sum_{0 \leq i \leq t} d_i q_i$.   By Lemma~\ref{lazy_len} we have $q_t+q_{t-1}-1 \leq n \leq q_{t+1}+q_{t}-2$.  Suppose that the elements of $A(n)$ are ordered by size and note that $q_t$ and $n$ are the least and greatest elements of $A(n)$ respectively.

By Lemma~\ref{min_period}, the minimal period of $Y_n$ is at least $q_t$, and clearly the maximal period of $Y_n$ is $n$.
Consequently, if there is some $p \in \PER(n)$ such that $p \notin A(n)$, then there are two consecutive periods $p_1, p_2 \in A(n)$ such that $p_1 < p < p_2$.  We find then that $Y_{n - p_1}$ has periods $p_2 - p_1$ and $p - p_1$. 

By the definition of $A(n)$, the period $p_1$ has the form $$p_1 = d_tq_t + d_{t-1}q_{t-1} + \cdots + d_{j+1}q_{j+1} + aq_j$$ for some $a \leq d_j$.  Hence $n - p_1$ has lazy representation
(possibly including some leading 0's) $(d_j-a) d_{j-1} \cdots d_0$.
By the induction hypothesis, we have $\PER(n-p_1) \subseteq A(n-p_1)$.
However, since $p_2$ and $p_1$ are consecutive periods of $Y_n$, we have $p_2 - p_1 = q_j$ if $a<d_j$
or $p_2 - p_1 = q_{j'}$, where $j'$ is the largest index $<j$ such that $d_{j'}>0$, if $a=d_j$.  By the
definition of $A(n - p_1)$, the least element of $A(n-p_1)$ is $q_j$ if $a<d_j$ or $q_{j'}$ if $a=d_j$.  It
follows that $p_2-p_1$ is the least element of $A(n-p_1)$. 
However, $p - p_1$ is smaller than $p_2 - p_1$, so we have $p-p_1 \in \PER(n-p_1)$ but $p-p_1 \notin A(n-p_1)$
which is a contradiction.
\end{proof}
Theorem~\ref{ostt} now follows from Lemmas~\ref{AsubPER} and \ref{PERsubA}.

Let us now apply these results to the infinite Fibonacci word ${\bf f} = 01001010\cdots$, which equals
the Sturmian characteristic word ${\bf x}_{\alpha}$ for
$\alpha = (3-\sqrt{5})/2 = [0,2,1,1,1,\ldots]$.  Recall that the $n$'th Fibonacci number is defined by
$F_0 = 0$, $F_1 = 1$, and $F_n = F_{n-1} + F_{n-2}$ for $n \geq 2$.
An easy induction shows that $q_i = F_{i+2}$ for
$i \geq 0$.   Here the ordinary Ostrowski representation
corresponds to the familiar and well-studied Fibonacci (or Zeckendorf) representation \cite{Lekkerkerker:1952,Zeckendorf:1972} as a sum of distinct Fibonacci numbers.   The lazy Ostrowski representation, on the other hand, corresponds to the so-called ``lazy Fibonacci representation'', as studied by Brown \cite{Brown:1965}.  This representation has the property that it contains no two consecutive $0$'s.

Theorem~\ref{ostt} now has the following implications for
the Fibonacci word.
\begin{corollary}
\leavevmode
\begin{enumerate}[(a)]
\item If the lazy Fibonacci representation of $n$ is
$n = F_{t_1} + F_{t_2} + \cdots + F_{t_r}$,
for $t_1 < t_2 < \cdots < t_r$, then the periods of
the length-$n$ prefix of the Fibonacci word are
$$F_{t_r},\ F_{t_r} + F_{t_{r-1}},\ F_{t_r} + F_{t_{r-1}} + F_{t_{r-2}},\ \ldots,\ 
F_{t_r} + F_{t_{r-1}} + \cdots + F_{t_1}.$$

\item 
The shortest prefix of $\bf f$ having exactly $n$ periods (including the trivial period) is of
length $F_{n+3} - 2$, for $n \geq 1$. 

\item 
The longest prefix of $\bf f$ having exactly $n$
periods (including the trivial period) is of 
length $F_{2n+2} - 1$, for $n \geq 1$.

\item 
The least period of ${\bf f}[0..m-1]$ is $F_n$
for $F_{n+1} - 1 \leq m \leq F_{n+2} - 2$ and
$n \geq 2$.
\end{enumerate}
\label{ten}
\end{corollary}

\begin{proof}
\leavevmode
\begin{enumerate}[(a)]
\item This is just a restatement of Theorem~\ref{ostt} for
the special case $\alpha = (3-\sqrt{5})/2$.

\item This corresponds to the lazy Fibonacci
representation $\overbrace{11\cdots 1}^n$, which
equals the sum $F_2 + F_3 + \cdots + F_{n+1}$,
for which a classical Fibonacci identity gives $F_{n+3} - 2$.

\item This corresponds to the lazy Fibonacci
representation $(10)^n$, which equals the sum
$F_3 + F_5 + \cdots + F_{2n+1}$, for which a classical
Fibonacci identity gives $F_{2n+2} - 1$.

\item Theorem~\ref{ostt} implies that the least period
of every $n$ with Ostrowski representation of length
$t$ is $F_{t+1}$.   Lemma~\ref{lazy_len} implies that
$q_{t-1} + q_{t-2} - 1 \leq n \leq q_t + q_{t-1} - 2$;
in other words, $F_{t+1} + F_{t} - 1 \leq n \leq F_{t+2} + F_{t+1} - 2$, or $F_{t+2} - 1 \leq n \leq F_{t+3} - 2$.   
\end{enumerate}
\end{proof}

For another connection between Ostrowski numeration and periods of Sturmian words,
see \cite{Schaeffer:2013}.  Saari \cite{Saari:2007} determined the least period of
every factor of the Fibonacci word, not just the prefixes; also see
\cite[Thm.~3.15]{Mousavi&Schaeffer&Shallit:2016}.

\section{Tightness of the period inequality}

Returning to our period inequality,
it is natural to wonder if the bound \eqref{bound1} is tight.
We exhibit a class of binary words for which it is.

Let $g_s$, for $s \geq 1$, be the prefix of length
$F_{s+2} - 2$ of $\bf f$.
Thus, for example, $g_1 = \epsilon$, $g_2 = 0$, $g_3 = 010$,
$g_4 = 010010$, and so forth. 
We now show that the bound \eqref{bound1} is tight, up to an 
additive factor, for the words $g_s$.  Let
$\tau = (1+\sqrt{5})/2$, the golden ratio.
\begin{theorem}
Take $x = g_s$ for $s \geq 4$.  Then the left-hand side of \eqref{bound1}
is $s-2$, while the right-hand side is
asymptotically $s+c$ for $c = 3 + \tau^2/2 - (\ln 2 \sqrt{5})/(\ln \tau)
\doteq 1.19632$.
\label{appr}
\end{theorem}

\begin{proof}
Take $x = g_s$.   By definition we have $n = |x| = F_{s+2} - 2$.    
By Corollary~\ref{ten} (b) we know that
$g_s$ has $s-1$ periods, and hence $s-2$ nontrivial periods.    Thus $\nnp(x) = s-2$.

Next let's compute $\ice(g_s)$.   Corollary~\ref{ten} (d) states that the least period of the prefix ${\bf f}[0..m-1]$ equals
$F_s$ for $F_{s+1} -1 \leq m \leq F_{s+2} - 2$, $s \geq 2$.  It follows
that the exponent of the prefix ${\bf f}[0..m-1]$ is
$m/F_s$ for $F_{s+1} -1 \leq m \leq F_{s+2} - 2$, $s \geq 2$. For
fixed $s$, the quantity $m/F_s$ is
maximized at $m = F_{s+2} - 2$, which gives an exponent of
$(F_{s+2} - 2)/F_s$.   It remains to see that the
sequence $((F_{s+2} - 2)/F_s)_{s \geq 2}$ is strictly increasing.
For this it suffices to show that
$(F_{s+2} - 2)/F_s < (F_{s+3}-2)/F_{s+1}$ for $s \geq 2$,
or, equivalently,
\begin{equation}
F_{s+2} F_{s+1} - F_s F_{s+3} < 2 F_{s+1} - 2 F_s.
\label{ffe}
\end{equation}
But an easy induction shows that the left-hand side of \eqref{ffe}
is $(-1)^s$, while the right-hand side is $2F_{s-1} \geq 2$.
Thus we see $e = \ice(g_s) = (F_{s+2} - 2)/F_s$.

Hence
the right-hand side of \eqref{bound1} is
$$ {{F_{s+2} - 2} \over {2F_s}} + 1 + 
{{\ln((F_{s+2} - 2)/2)} \over {\ln( {{F_{s+2} - 2} \over {F_{s+1} - 2}} ) } } .$$
Now
use the Binet formula for Fibonacci numbers, which implies
that $F_s \sim \tau^s/\sqrt{5}$, and the fact that
$\lim_{s \rightarrow \infty} F_s/F_{s-1} = \tau$,
to obtain that the right-hand side of \eqref{bound1}
is asymptotically 
$$ {{\tau^2} \over 2} + 1 + (s+2) - (\ln 2\sqrt{5})/(\ln \tau). $$
This gives the desired result.
\end{proof}

\section{Two measures of periodicity}
\label{twomeas}

Corollary~\ref{bound2} suggests that the quantity
$$ M(x) := {{\nnp(x)} \over {\ice(x) \ln |x| }} $$
is a measure of periodicity for finite
words $x$.  It also suggests studying
the following measures of periodicity
for infinite words $\bf x$.  For $n \geq 2$ let
$Y_n$ be the prefix of length $n$ of $\bf x$.  Then define
\begin{align*}
 P({\bf x}) := \limsup_{n \rightarrow \infty}  \ 
M(Y_n)  \\
p({\bf x}) := \liminf_{n \rightarrow \infty}  \ 
M(Y_n)  
\end{align*}
From Theorem~\ref{expbord}, 
we know that for the ``typical'' infinite word $\bf x$ we have
$P({\bf x}) = p({\bf x}) = 0$.   Thus it is of interest to find words $\bf x$
where $P({\bf x})$ and $p({\bf x})$ are
large.
In this section we compute these measures for several infinite words.
\begin{theorem}
Let $\bf f$ denote the Fibonacci infinite word.  Then
$P({\bf f}) = 1/(\tau^2\ln \tau) \doteq 0.79375857$ and
$p({\bf f}) = 1/(2 \tau^2\ln \tau) \doteq 0.396879286$.
\label{fibo2}
\end{theorem}
\begin{proof}
This follows immediately from Corollary~\ref{ten},
together with the calculation of $\ice$ given in
the proof of Theorem~\ref{appr}.
\end{proof}

The {\it period-doubling word\/}
{\bf d} is defined to be the fixed
point of the morphism sending
$1 \rightarrow 10$ and $0\rightarrow 11$;
see \cite{Damanik:2000}.
\begin{theorem}
$P({\bf d}) = {1 \over {2 \ln 2}} \doteq 0.7213$  and
$p({\bf d}) = {1 \over {4 \ln 2}} \doteq 0.36067$.
\end{theorem}
\begin{proof}
Since $\bf d$ is not a Sturmian word, or even closely related to one, we need to use different techniques from those we used previously.

Let $r(n)$ denote the number of periods (including the trivial period) in the
length-$n$ prefix of $\bf d$.   We use $(n)_2$ to denote the
canonical base-$2$ representation of $n$, and $(n,p)_2$ to denote the base-$2$ representation of $n$ and $p$ as a sequence of pairs of bits (where the shorter representation is padded with leading zeros, if necessary).

We can use the theorem-proving software {\tt Walnut} to calculate the periods of prefixes of $\bf d$.   (For more about {\tt Walnut}, see \cite{Mousavi:2016}.)  We sketch the ideas briefly.   

We can write a first-order logical formula $\ \pdp(m,p) \ $ stating that
the prefix of length $m\geq 1$ of $\bf d$ has period $p$,
$1 \leq p \leq m$:
\begin{align*}
 \pdp(m,p) & := (1 \leq p \leq m) \ \wedge\ {\bf d}[0..m-p-1] = 
{\bf d}[p..m-1] \\
&= (1 \leq p \leq m) \ \wedge\ \forall t \ (0 \leq t<m-p) \implies {\bf d}[t] = {\bf d}[t+p] .
\end{align*}
Such a formula can be automatically
translated, using {\tt Walnut}, to an
automaton that recognizes the language 
$$\{ (n,p)_2 \suchthat \text{ the length-$n$ prefix of $\bf d$ has period $p$} \}.$$
We depict it below.
% def pdp "(p>=1) & (m>=1) & (p<=m) & (At (t+p<m) => PD[t] = PD[t+p])":
\begin{center}
\includegraphics[width=6.5in]{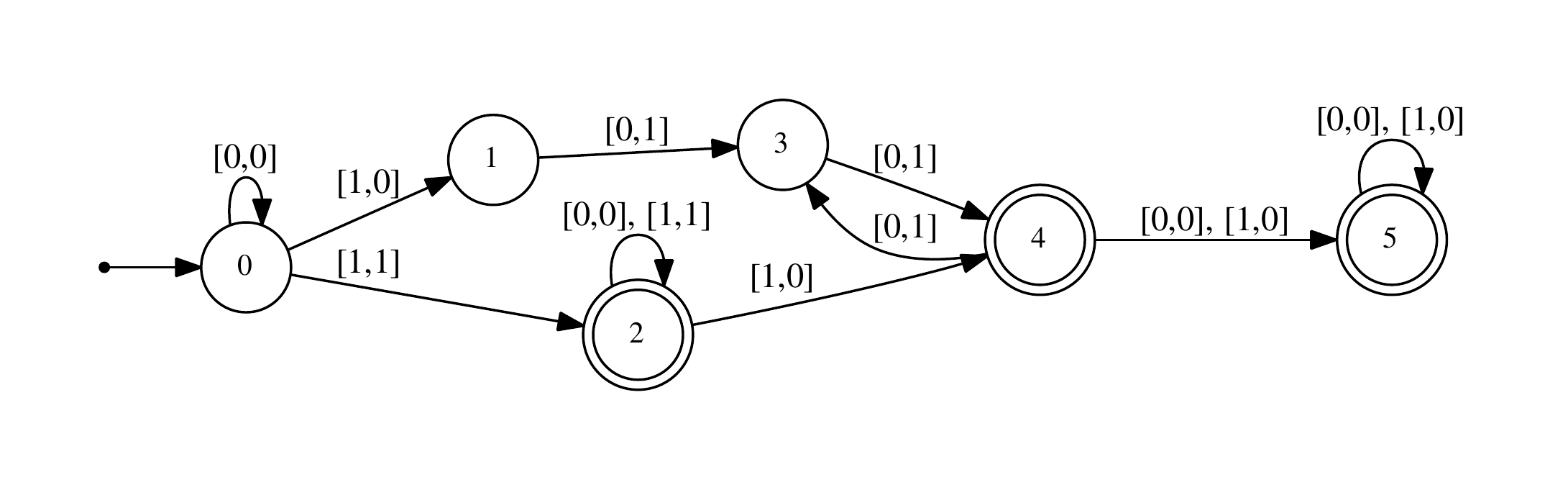}
\end{center}
Such an automaton can be automatically converted by {\tt Walnut} to a linear
representation for $r(n)$, as discussed in
\cite{Charlier&Rampersad&Shallit:2012}.  This
is a triple $(v, \rho, w)$ where $v, w$ are
vectors, and $\rho$ is a matrix-valued morphism,
such that $r(n) = v \cdot \rho ( (n)_2 ) \cdot w$.
The values are given below:
$$ v = [1\, 0\, 0\, 0\, 0\, 0] \quad
\rho(0) = \left[ \begin{array}{cccccc}1&0&0&0&0&0\\
0&0&0&1&0&0\\
0&0&1&0&0&0\\
0&0&0&0&1&0\\
0&0&0&1&0&1\\
0&0&0&0&0&1
\end{array} \right] \quad
\rho(1) = \left[ \begin{array}{cccccc} 0&1&1&0&0&0\\
0&0&0&0&0&0\\
0&0&1&0&1&0\\
0&0&0&0&0&0\\
0&0&0&0&0&1\\
0&0&0&0&0&1
\end{array} \right]
\quad
w = \left[ \begin{array}{c}
0\\
0\\
1\\
0\\
1\\
1
\end{array} \right]
. $$
From this, using the technique described in
\cite{Goc&Mousavi&Shallit:2013},
we can easily compute the relations
\begin{align*}
r(0) &= 0 \\
r(2n+1) &= r(n) + 1, \quad n \geq 0 \\
r(4n) &= r(n) + 1, \quad n \geq 1 \\
r(4n+2) &= r(n) + 1, \quad n \geq 0.
\end{align*}
Reinterpreting this definition for $r$, we see that
$r(n)$ is equal to the length of the (unique) factorization
of $(n)_2$ into the factors $1$, $00$, and $10$.
It now follows that 
\begin{enumerate}[(a)]
\item The smallest $m$ such that $r(m) = n$ is $m = 2^n - 1$;
\item The largest $m$ such that $r(m) = n$ is $m = \lfloor 2^{2n+1}/3 \rfloor$, with $(m)_2 = (10)^n$.
\end{enumerate}

Similarly, we can use {\tt Walnut} to determine the smallest period $p$ of every
length-$n$ prefix of $\bf d$.   We use the predicate
% def pdlp "$pdp(n,p) & (Aq (q>=1&q<p) => ~$pdp(n,q))":
$$ \pdlp(n,p) := \pdp(n,p) \ \wedge \ 
\forall q\ (1 \leq q < p) \implies \pdp(n,q).$$
This gives the automaton
\begin{center}
\includegraphics[width=5.5in]{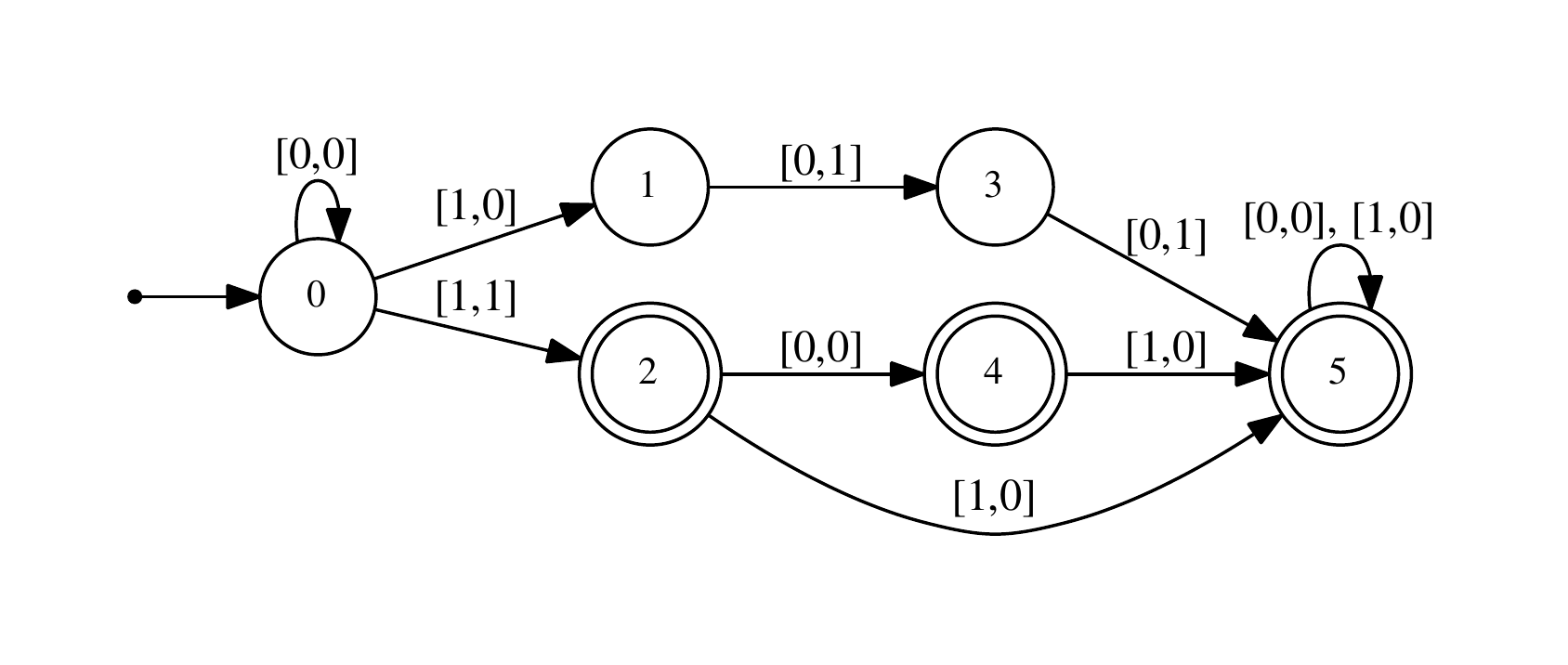}
\end{center}
Inspection of this automaton
shows that least period of the prefix of length
$n$ is, for $s \geq 2$, equal to
$3 \cdot 2^{s-2}$ for $2^s \leq n < 5 \cdot 2^{s-2}$ and $2^s$ for
$5 \cdot 2^{s-2} \leq n < 2^{s+1}$.   It follows that 
the 
initial critical exponent of every prefix of $\bf d$
of length $n$, for $2^t - 1 \leq n \leq 2^{t+1} - 2$,
is $2- 2^{1-t}$.   

The result now follows.
\end{proof}

\begin{theorem}
Let ${\bf t} = t_0 t_1 t_2 \cdots = 01101001\cdots$ be the Thue-Morse word, the fixed
point of the morphism $\mu$ described above.
Then
$P({\bf t}) = 3/(10 \ln 2) \doteq 0.4328$ 
and $p({\bf t}) = 0$.
\end{theorem}
\begin{proof}
We have $\ice(x) = 5/3$ for every prefix
$x$ of $\bf t$ of length $\geq 5$, a claim that can
easily be verified with {\tt Walnut}.

For the value of $p({\bf t})$, it suffices to 
observe
that $\nnp(x) = 1$ if $x$ is a prefix
of $\bf t$ of length $3\cdot 2^n + 1$
for $n \geq 0$, which can also be verified with 
{\tt Walnut}.

For $P({\bf t})$ it suffices to show that the shortest prefix of $\bf t$
having $n$ nontrivial periods is of 
length $2^{2n-1} + 2$.    For this we can use {\tt Walnut}, but the analysis is somewhat complicated.   Letting
$v(n)$ denote the number of nontrivial periods of the
length-$n$ prefix of $\bf t$, we can mimic what we did for the period-doubling word, obtaining the matrices and the following relations for $n \geq 0$:
\begin{align*}
v(4n) &= v(n) + [n \not= 0] \\
v(4n+3) &= v(4n+1) \\
v(8n+1) &= v(2n+1) + t_n \\
v(8n+2) &= v(2n+1) + t_n\\
v(8n+6) &= v(4n+1) + 1-t_n \\
v(16n+5) &= v(2n+1) + 1 \\
v(16n+13) &= v(4n+1) + 1 .
\end{align*}
Here $[n \not= 0]$ is the Iverson bracket, which evaluates
to $1$ if the condition holds and $0$ otherwise.

Now a tedious induction on $m$, which we omit, shows that
\begin{align*}
m \text{ is even and } v(m) \geq n & \implies m \geq 2^{2n-3} + 2; \\
m \text{ is odd and } v(m) \geq n & \implies m \geq 2^{2n-2} + 1, 
\end{align*}
and furthermore $v(2^{2n-3} + 2) = n$ for $n \geq 2$.
It follows that the shortest prefix of $\bf t$ 
having $n$ nontrivial periods is of length
$2^{2n-1} + 2$ for $n \geq 2$, from which the desired result follows.
\end{proof}

\begin{remark}
The {\tt Walnut} commands for the last two results are available on the third author's web page, at

\centerline{\url{https://cs.uwaterloo.ca/~shallit/papers.html} \ .}

\noindent {\tt Walnut} itself is available at

\centerline{  \url{https://github.com/hamousavi/Walnut} \ .}
\end{remark}

\begin{remark}
It would be interesting to compute
the values of
\begin{align*}
D_1 := \inf_{n \geq 1} \ \sup_{x \in \{ 0, 1 \}^n}\ M(x) \\
    D_2 := \liminf_{n \rightarrow \infty} \ \sup_{x \in \{ 0, 1 \}^n}\ M(x).
\end{align*}
Theorem~\ref{fibo2} shows that $D_2 \geq
1/(2 \tau^2 \ln \tau) \doteq 0.396879286$.
Thus, for example, for every sufficiently large $n$ there is a length-$n$ binary string $x$ with $M(x) \geq .396$.
\end{remark}

\section{Shortest overlap-free binary word with $p$ periods}

In this section and the following one, we consider how quickly
the number of periods can grow if we enforce an upper bound
on the exponent of repetitions occurring in the word.
% critical exponent of a word.   Recall that the
% {\it critical exponent\/} $\ce(x)$ of a finite or infinite word $x$
% is defined to be
% $$ \ce(x) := \sup_{{p \text{ a nonempty}} \atop {\text{ factor of $x$}}} \exp(p).$$

Recall that an {\it overlap\/} is a word of the form $axaxa$,
where $a$ is a single letter and $x$ is a (possibly empty) word.
An example in English is the word {\tt alfalfa}.  We say a
word is {\it overlap-free} if no finite factor is an overlap.

Define $f(p)$ to be the length of the shortest binary overlap-free word
having $p$ nontrivial periods.     Recall that we call a border $w$
of $x$ {\it short\/} if $|w| < |x|/2$.

Define the morphism $\mu$ by $\mu(0) = 01$ and $\mu(1) = 10$.
If $w = axa$ for a single letter $a$ and (possibly empty) word
$x$, define $\gamma(w) = a^{-1} \mu^2 (w) a^{-1}$, or, in other words,
the word $\mu^2(w)$ with an $a$ removed from the front and back.

\begin{lemma}
Define a sequence of words $(A_n)_{n \geq 3}$ as follows:
$$
A_n = \begin{cases}
        001001100100, & \text{if $n = 3$}; \\
        \gamma(A_{n-1}), & \text{if $n \geq 4$}.
        \end{cases}
$$
%For example,
%\begin{align*}
%A_3 &= 001001100100 \\
%A_4 &= 1100110100101100110100110010110011010010110011 \\
%A_5 &= 
%\text{\tiny\rm 00110010110011010011001011010010110011010010110100110010110011010011001011010010110011010011001011001101001011010011001011001101001100101101001011001101001011010011001011001101001100 }
%\end{align*}
Then $A_n$ is a palindrome with $n$ short palindromic borders for $n \geq 3$.
\end{lemma}

\begin{proof}

Observe that if $w$ is a palindrome, then so is $\gamma(w)$.
Write $\overline{a} = 1-a$ for $a \in \{0,1\}$.

We now prove the claim by induction on $n$.
It is true for $n = 3$, since the borders are $0, 00, $ and $00100$.

Now assume the result is true for $n$; we prove it for $n+1$.  
Suppose $n$ short palindromic borders of $A_n$ are
$w_1, w_2, \ldots, w_n$, and each starts with the letter
$a$.  From the observation above, we know that
$A_{n+1} = \gamma(A_n)$ is a palindrome.   We claim that 
$\overline{a}, \gamma(w_1), \gamma(w_2), \ldots, \gamma(w_n)$ are short palindromic borders of
$\gamma(A_n)$.  

To see that $\overline{a}$ is a border of $A_{n+1}$, note that
$A_n = awa$ for some $w$, so $\gamma(A_n) = \overline{a} \overline{a} a \mu^2(w) 
a \overline{a}\overline{a}$.

Otherwise, let $w_i$ be a palindromic border of $A_n$.   Since it is
short, we have $A_n = w_i y w_i$ for some $y$.  Then $\gamma(w_i)$ is both a prefix
and suffix of $\gamma(A_n)$ and hence is a palindromic border of $A_{n+1}$.
The claim about the length of the borders is trivial.

Thus $A_{n+1}$ has at least $n+1$ palindromic short borders.
\end{proof}

\begin{corollary}
We have $f(1) = 2$, $f(2) = 5$, and
$f(p) \leq (17/6) 4^{p-2} + 2/3$ for $p \geq 3$.
\end{corollary}

\begin{proof}
For $p = 1$, the shortest binary overlap-free word with $1$ nontrivial
period is $00$.   For $p = 2$ it is $00100$.  

Next we argue, by induction on $p$, that 
that each $A_p$, for $p \geq 3$, is overlap-free.   
The base case is $p = 3$, and is easy to check.
Otherwise assume the result is true for $A_p$.
We now use a classical result
that if a word $x$ is overlap-free, then so is $\mu(x)$ \cite{Thue:1912}.
Applying this twice, we see that $\mu^2(A_p)$ is overlap-free.
Then $A_{p+1} = \gamma(A_p)$ is overlap-free, since it is a factor of $\mu^2(A_p)$.

As we have seen above, $A_p$ has $p$ borders and hence $p$ nontrivial
periods.  The only thing left to verify is that
$|A_p| = (17/6) 4^{p-2} + 2/3$ for $p \geq 3$.   This is an
easy induction, and is left to the reader.
\end{proof}

\begin{remark}
One can go from $A_p$ to $A_{p+1}$, for $p \geq 3$,
via the following procedure, which we state without proof.
Write $A_p$ in terms of its
run-length encoding, that is,
$A_p = a^{e_1} b^{e_2} a^{e_3} b^{e_4} \cdots $,
where $a \not= b$ and all the $e_i$ are positive.
Then, considering $c^e$ as the pair $(c,e)$, apply the following morphism:
\begin{align*}
(0,1) & \rightarrow 1101 \\
(1,1) & \rightarrow 0010 \\
(0,2) & \rightarrow 11001101 \\
(1,1) & \rightarrow 00110010
\end{align*}
Finally, drop the last two symbols.
\end{remark}

\begin{remark}
We conjecture that the words $A_p$ constructed above are actually
the shortest overlap-free binary words with $p$ periods with $p \geq 3$, but we
do not currently have a proof of this claim in general.  The sequence
$(f(p))$ is sequence \seqnum{A334811} in the
{\it On-Line Encyclopedia of Integer Sequences} \cite{Sloane:2020}.
\end{remark}

\section{Shortest squarefree ternary word with $p$ periods}

Recall that a {\it square\/} is a nonempty word of the form $xx$,
such as the English word {\tt murmur}.  A word is {\it squarefree\/}
if no finite factor is a square.

Let $g(p)$ be the length of the shortest ternary squarefree word
having $p$ nontrivial periods.   Here are the first few values
of $g$, computed through exhaustive search.
\begin{center}
\begin{tabular}{c|ccccc}
$p$ & 0 & 1 & 2 & 3 & 4 \\
\hline
$g(p)$ & 1 & 3 & 7 & 23 & 59
\end{tabular}
\end{center}

\begin{theorem}
For $p \geq 3$ we have $g(p) \leq {{17}\over {12}} 4^{p-1} + 1/3$.
\end{theorem}

\begin{proof}
Consider the words $A_p$ defined above.   Suppose $A_p$ starts and
ends with the letter $a$.   Let $B_p$ be the word whose
$i$'th letter is the number of occurrences of $\overline{a}$
between the $i$'th and the $(i+1)$'th occurrence of $a$.
For example, we have
\begin{align*}
B_3 &= 0102010 \\
B_4 &= 02012102012021020121020 \\
B_5 &=  \text{\tiny\rm 0201202102012101202101210201202102012101202102012021012102012021020121012021012102012021020 }
\end{align*}
Then each $B_p$ is squarefree.  For if $B_p$ had a square,
say $c_1 c_2 \cdots c_t c_1 c_2 \cdots c_t$, then
$A_p$ has the overlap 
$$a b^{c_1} a b^{c_2} \cdots a b^{c_t} a b^{c_1}
a b^{c_2} \cdots a b^{c_t} a,$$
where $b = \overline{a}$, a contradiction.

Furthermore, each border of $A_p$, except the border of length
$1$, corresponds via this map to a border of $B_p$.  So
$\nnp(B_p) = p-1$.   By induction we can show
$|A_p| = |B_p|/2 = (17/12) 4^{p-2} + 1/3$ for $p \geq 4$.
It follows that $g(p) \leq (17/12) 4^{p-1} + 1/3$.
\end{proof}

\begin{remark}
Our bound is clearly not optimal.   It would be interesting to
obtain better bounds for $g(p)$.  The sequence
$(g(p))$ is sequence \seqnum{A332866} in the
{\it On-Line Encyclopedia of Integer Sequences} \cite{Sloane:2020}.
\end{remark}

\begin{remark}
One can go from $B_p$ to $B_{p+1}$, for $p \geq 4$, using the following
procedure, which we state without proof.
Take $B_p$ and replace every other $1$ in it with $3$.
Then apply the following morphism:
\begin{align*}
0 &\rightarrow 0201 \\
1 &\rightarrow 2101 \\
2 &\rightarrow 2021 \\
3 &\rightarrow 0121 .
\end{align*}
Finally, drop the last letter.
\end{remark}

\section*{Acknowledgments}

Thanks to Anna Frid, Jean-Paul Allouche,
Luke Schaeffer, Kalle Saari, \v{S}t\v{e}p\'an Holub, Jean Berstel, and Val{\'e}rie Berth\'e for their helpful comments.

\end{document}